\documentclass[a4paper,UKenglish,cleveref, autoref, thm-restate]{lipics-v2021}

\hideLIPIcs  

\title{Vertex Cover Interdiction in Bipartite Graphs} 

\author{Takehiro Ito}{Graduate School of Information Sciences, Tohoku University, Sendai, Japan}{takehiro@tohoku.ac.jp}{https://orcid.org/0000-0002-9912-6898}{JSPS KAKENHI Grant Numbers JP24H00686, JP24H00690}
\author{Naonori Kakimura}{Faculty of Science and Technology, Keio University, Yokohama, Japan}{kakimura@math.keio.ac.jp}{https://orcid.org/0000-0002-3918-3479}{JSPS KAKENHI Grant Numbers JP23K21646, JP26K02867, JP26K21777 and JST ERATO Grant Number JPMJER2301}
\author{Naoyuki Kamiyama}{Graduate School of Informatics, Kyoto University, Kyoto, Japan}{kamiyama@amp.i.kyoto-u.ac.jp}{https://orcid.org/0000-0002-7712-2730}{JSPS KAKENHI Grant Number JP24K14825}
\author{Yusuke Kobayashi}{Research Institute for Mathematical Sciences, Kyoto University, Kyoto, Japan}{yusuke@kurims.kyoto-u.ac.jp}{https://orcid.org/0000-0001-9478-7307}{JSPS KAKENHI Grant Numbers JP24K02901, 26K21777}
\author{Yoshio Okamoto}{Graduate School of Informatics and Engineering, The University of Electro-Communications, Chofu, Japan}{okamotoy@uec.ac.jp}{https://orcid.org/0000-0002-9826-7074}{JSPS KAKENHI Grant Numbers JP23K10982, JP26K23806 and JST ERATO Grant Number JPMJER2301}

\authorrunning{T. Ito, N. Kakimura, N. Kamiyama, Y. Kobayashi, and Y. Okamoto} 

\Copyright{Takehiro Ito, Naonori Kakimura, Naoyuki Kamiyama, Yusuke Kobayashi, and Yoshio Okamoto} 



\nolinenumbers 

\EventEditors{Lin Chen and Nicole Megow}
\EventNoEds{2}
\EventLongTitle{37th International Symposium on Algorithms and Computation (ISAAC 2026)}
\EventShortTitle{ISAAC 2026}
\EventAcronym{ISAAC}
\EventYear{2026}
\EventDate{December 6--9, 2026}
\EventLocation{Hangzhou, China}
\EventLogo{}
\SeriesVolume{399}
\ArticleNo{8}

\usepackage{amsmath,amsfonts,latexsym,amssymb,amsthm}
\usepackage{graphicx} 
\usepackage{appendix}
\usepackage{framed}
\usepackage{xspace}
\usepackage{tcolorbox}
\usepackage{comment}
\usepackage{hyperref}
\usepackage{cleveref}

\graphicspath{{./figure/}}

\usepackage{complexity}

\newcommand{\bvci}{\textsc{Bipartite Vertex Cover Interdiction}\xspace}
\newcommand{\bminvci}{\textsc{Bipartite Minimum Vertex Cover Interdiction}\xspace}
\newcommand{\vci}{\textsc{Vertex Cover Interdiction}\xspace}
\newcommand{\minvci}{\textsc{Minimum Vertex Cover Interdiction}\xspace}

\newcommand{\minvcfamily}{\mathcal{C}^\ast}

\newcommand{\red}[1]{#1}

\begin{document}

\maketitle

\begin{abstract}
    In the vertex cover interdiction problem, we are given an undirected graph $G=(V,E)$, two integers $t$ and $k$ and a vertex subset $B\subseteq V$, and we are asked to find a set $X \subseteq B$ with $|X|\leq t$ such that $X$ hits (i.e., intersects) all the vertex covers of $G$ of size at most $k$.
    Recently, Gr\"une and Wulf proved that the problem is $\Sigma_2^p$-complete. However, their reduction relied on the fact that the vertex cover problem is \NP-complete.
    This, in turn, means that we do not know the complexity status of the vertex cover interdiction problem when the input graph is restricted to a bipartite graph since the vertex cover problem can be solved in polynomial time for bipartite graphs.
    One of our main results shows that the vertex cover interdiction problem is \NP-complete for bipartite graphs. In contrast, when $k$ is restricted to the minimum vertex cover size, i.e., we are only required to hit all the minimum vertex covers, we show that the vertex cover interdiction problem can be solved in polynomial time for bipartite graphs. 
    This motivates us to study the parameterized complexity of the vertex cover interdiction problem for bipartite graphs when the difference of $k$ and the minimum vertex cover size is taken as a parameter. With this parameter, we show that the problem is $\W[1]$-hard, but can be solved in polynomial time when the parameter is constant (i.e., in \XP~time).
    We also show that the problem is fixed-parameter tractable when parameterized by $k$. 
\end{abstract}


\section{Introduction}

In the vertex cover interdiction problem, we are given an undirected graph $G=(V,E)$, two integers $t$ and $k$, and a vertex subset $B\subseteq V$. We are asked to find a set $X \subseteq B$ with $|X|\leq t$ such that $X$ hits all the vertex covers of $G$ of size at most $k$;  
see \Cref{fig:example} for an example. 
This problem is a mathematical model of the following situation.
We are operating a network that is represented by a graph $G$ and want to monitor the links in $G$ by placing a certain number of sensors.
The placement of a sensor is constrained to a vertex of $G$, and the aggregated placement of sensors corresponds to a vertex cover of $G$ since we want to monitor all the links (i.e., edges).
 Due to severe budget restrictions, the number of sensors is limited to $k$.
Now, we want to evaluate the network $G$ based on its tolerance against attacks by intruders.
Intruders may destroy a sensor, but they have no knowledge about the placement of sensors.
Therefore, they want to find a small subset $X$ (i.e., a subset of size at most $t$) of the vertices such that by observing only those vertices, they are certain about the placement of at least one sensor.
The observation by intruders is restricted to some subset $B$ of the vertices since the network $G$ is not fully exposed to them.
When such a set $X$ exists, the size bound $t$ is considered a measure of the tolerance of the network $G$.

\begin{figure}[tbp]
    \centering
    \includegraphics[width=0.75\linewidth]{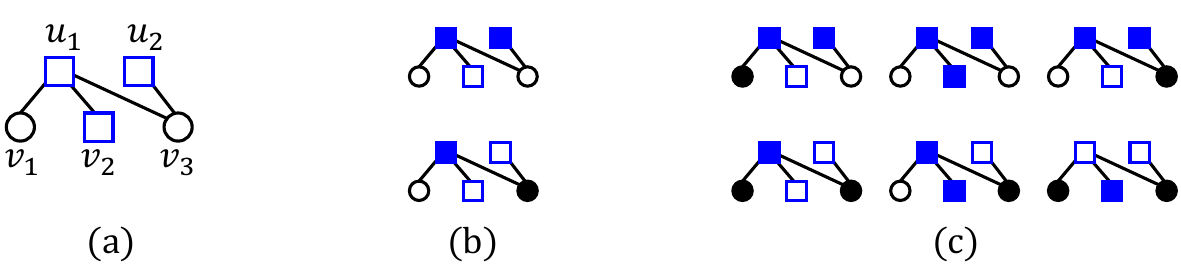} 
    \caption{(a) Graph $G$ with $B= \{u_1, u_2, v_2\}$, (b) all vertex covers of size two in $G$, and (c) all vertex covers of size three in $G$;
    in (b) and (c), the vertices included in each vertex cover are colored solid. 
    If $t=1$ and $k=2$, then $\{u_1\}$ is the unique solution because $u_1$ is the only vertex in $B$ that hits (intersects) all vertex covers of size at most two. 
    On the other hand, if $t=1$ and $k=3$, then there is no feasible solution because no single vertex in $B$ hits all vertex covers of size at most three.}
    \label{fig:example}
\end{figure}

Gr\"une and Wulf \cite{DBLP:conf/ipco/GruneW25} proved that the vertex cover interdiction problem is $\Sigma_2^p$-complete. 
The complexity class $\Sigma_2^p$ contains the well-known complexity class \NP \ (i.e., $\NP \subseteq \Sigma_2^p$), but it is an open problem whether proper inclusion holds or not.
A typical $\Sigma_2^p$-complete problem is the $\exists\forall$-SAT, in which we are given a CNF formula $\varphi(X,Y)$, where $X$ and $Y$ are sets of variables, and we want to determine whether there exists an assignment to $X$ such that for all assignments to $Y$, the expression $\varphi(X,Y)$ is true.
Indeed, Gr\"une and Wulf  \cite{DBLP:conf/ipco/GruneW25} provided a general framework to prove that if the base problem $\Pi$ is \NP-complete and it has a particular property that is common to many problems, then the interdiction problem derived from $\Pi$ is $\Sigma_2^p$-complete.
For the vertex cover interdiction problem, the base problem is the vertex cover problem, in which we are given an undirected graph $G=(V,E)$ and a non-negative integer $k$, and we want to determine whether $G$ contains a vertex cover of size at most $k$.
The vertex cover problem is one of the oldest problems that have been known to be \NP-complete~\cite{DBLP:conf/coco/Karp72}.

This, in turn, means that we do not know the complexity status of the vertex cover interdiction problem when the input graph is restricted to a bipartite graph since the vertex cover problem can be solved in polynomial time for bipartite graphs. Our research question boils down to the following sentence:
    Determine the computational complexity of the vertex cover interdiction problem when the input graph is restricted to a bipartite graph.

One of our main results shows that the vertex cover interdiction problem is \NP-complete for bipartite graphs. In contrast to the \NP-completeness above, when $k$ is restricted to the minimum vertex cover size, i.e., we are only required to hit all the minimum vertex covers, the vertex cover interdiction problem can be solved in polynomial time for bipartite graphs. 
This motivates us to study the parameterized complexity of the vertex cover interdiction problem for bipartite graphs when the difference of $k$ and the minimum vertex cover size is taken as a parameter. With this parameter, we show that the problem is $\W[1]$-hard, but can be solved in polynomial time when the parameter is constant (i.e., in \XP~time).
In contrast, we also show that \bvci is fixed-parameter tractable when parameterized by $k$.

\paragraph*{Related Results}
Recently, the interdiction problems on graphs and networks have been actively studied.
Surveys by Smith and Song~\cite{DBLP:journals/eor/SmithS20} and by Ausiello et al.~\cite{DBLP:journals/csr/AusielloBFLR26} gave general overviews and include pointers to existing work on the network interdiction problems.
Some theory papers are devoted to approximation algorithms (e.g.~\cite{DBLP:journals/mor/ChestnutZ17,DBLP:conf/icalp/ChenWZ22}) and fixed-parameter tractability (e.g.~\cite{DBLP:journals/networks/BazganFNNS19,DBLP:journals/algorithmica/DvorakK18}).

The work by Bazgan, Toubaline, and Tuza~\cite{DBLP:journals/dam/BazganTT11} is the most closely related to our work.
Their paper studied, among other results, the following problem related to vertex covers, which we phrase as a decision problem.
We are given an undirected graph $G=(V,E)$ with non-negative integer vertex weight and integers $t, u$. Then, we want to find a set $X$ of at most $t$ vertices such that the minimum vertex cover weight of $G-X$ is at most $u$.
They showed that the problem is \NP-hard even for bipartite graphs, but polynomial-time solvable for unweighted bipartite graphs.
We emphasize that their problem is different from ours in this paper since our problem has a restriction $B \subseteq V$ on the selection of our hitting sets.
Similar studies have been done by Costa, de Werra, and Picouleau~\cite{DBLP:journals/jco/CostaWP11} and Bentz, Costa, Picouleau, Ries, and de Werra~\cite{DBLP:journals/jda/BentzCPRW12} who considered hitting every minimum vertex cover at $d$ vertices.

\section{Preliminaries}

The set of all non-negative integers is denoted by $\mathbb{Z}_+$.
For two sets $X$ and $Y$, their symmetric difference $(X \setminus Y) \cup (Y \setminus X)$ is denoted by $X \mathbin{\triangle} Y$. 

All graphs in this paper are finite, undirected, and simple.
A graph is denoted as a pair $G=(V,E)$ of its vertex set $V$ and its edge set $E$.
When we talk about a \emph{bipartite} graph, the vertex set $V$ is assumed to be partitioned into two sets $V_1$ and $V_2$ such that every edge $\{v_1,v_2\}$ satisfies $v_1 \in V_1$ and $v_2 \in V_2$; namely, there is no edge between two vertices of $V_1$ or two vertices of $V_2$.

Let $G=(V,E)$ be a graph.
For a vertex subset $X \subseteq V$, we denote by $N_G(X)$ the (open) neighborhood of $X$; i.e., $N_G(X) = \{v \in V \mid \{u,v\} \in E \text{ for some } u \in X\} \setminus X$.
In addition, we denote by $N_G[X]$ the closed neighborhood of $X$; i.e., $N_G[X] = N_G(X) \cup X$.
For a vertex $v \in V$, we denote by $\delta_G(v)$ the set of edges incident to $v$; i.e., $\delta_G(v) = \{e \in E \mid v \in e\}$.

A \emph{vertex cover} of a graph $G=(V,E)$ is a vertex subset $C\subseteq V$ such that $C\cap \{u,v\} \neq \emptyset$ for all edges $\{u,v\} \in E$.
A \emph{minimum vertex cover} of $G$ is a vertex cover of minimum size.
We denote the family of vertex covers of $G$ of size at most $k$ by $\mathcal{C}(G, k)$.  

A \emph{matching} of a graph $G=(V,E)$ is an edge subset $M \subseteq E$ such that $e \cap e' =\emptyset$ for any two edges $e, e' \in M$.
A \emph{maximum matching} of $G$ is a matching of maximum size.
It is known that in every bipartite graph $G$, the size of a maximum matching is equal to the size of a minimum vertex cover, which we denote by $\mu(G)$.

A set $X$ \emph{hits} a set $A$ if $X\cap A \neq \emptyset$.
A set $X$ \emph{hits} a set family $\mathcal{A}$ if $X\cap A \neq \emptyset$ for all $A \in \mathcal{A}$.
In this case, we often say that $X$ is a \emph{hitting set} (or a \emph{transversal}) of $\mathcal{A}$.

A \emph{dominating set} of a graph $G=(V,E)$ is a vertex subset $D \subseteq V$ such that every vertex $v \in V \setminus D$ is adjacent to at least one vertex in $D$.
Equivalently, $D$ is a hitting set for the family of the closed neighborhoods $\{N_G[\{v\}] \mid v \in V \}$.

The main focus of this paper is the following problem. 
\begin{tcolorbox}[colback=white,sharp corners]
\begin{description}
\item[Problem:] \bvci
\item[Input:] A bipartite graph $G=(V, E)$, two non-negative integers $t,k$, and a vertex subset $B\subseteq V$
\item[Question:] Does there exist a set $X \subseteq B$ such that $|X|\leq t$ and $X\cap C\neq \emptyset$ for every vertex cover $C$ of $G$ of size at most $k$?
\end{description}
\end{tcolorbox}
In other words, a solution $X$ to \bvci is a hitting set for the family $\mathcal{C}(G,k)$ of vertex covers of $G$ of size at most $k$.

When $k$ equals the size of a minimum vertex cover of $G$,  
the problem is referred to as \bminvci. 
Note that, as we will see in Appendix~\ref{sec:sigma2pmin}, 
this problem becomes $\Sigma_2^p$-complete when $G$ can be a non-bipartite graph. 

\begin{tcolorbox}[colback=white,sharp corners]
\begin{description}
\item[Problem:] \bminvci
\item[Input:] A bipartite graph $G=(V, E)$, a non-negative integer $t$, and a vertex subset $B\subseteq V$
\item[Question:] Does there exist a set $X \subseteq B$ such that $|X|\leq t$ and $X\cap C\neq \emptyset$ for every minimum vertex cover $C$ of $G$?
\end{description}
\end{tcolorbox}

\section{NP-Completeness of \bvci}

In this section, we prove that \bvci\ is \NP-complete, 
whereas the problem in general graphs is $\Sigma_2^p$-complete~\cite{DBLP:conf/ipco/GruneW25}.

\begin{proposition}\label{prop:np}
    \bvci is in the class \NP. 
\end{proposition}
\begin{proof}
    Let $(G, t, k, B)$ be an instance of \bvci, where $G = (V,E)$.
    It suffices to show that we can determine in polynomial time whether a given vertex set $X \subseteq V$ is a solution to the instance.
    First, we can trivially check whether $X \subseteq B$ and $|X| \leq t$ in polynomial time.
    Next, to check if $X \cap C \neq \emptyset$ for every vertex cover $C$ of size at most $k$, we assign a weight of $1$ to each vertex in $V \setminus X$, and a sufficiently large weight, say $k+1$, to each vertex in $X$.
    Under this weight assignment, any vertex cover containing at least one vertex from $X$ has a weight strictly greater than $k$. 
    Then, $X$ satisfies the intersection condition if and only if the minimum weight of any vertex cover of $G$ is strictly greater than $k$. 
    Since we can find a minimum weight vertex cover of a bipartite graph in polynomial time (see, e.g., \cite[Chapter 17]{Sch03}), the intersection condition can be checked in polynomial time.
\end{proof}

To prove the \NP-hardness of \bvci, we construct a reduction from the \textsc{Dominating Set} problem, where we are given a graph $G'$ and a positive integer $t'$, and are asked to determine whether $G'$ has a dominating set of size at most $t'$. 
\textsc{Dominating Set} is known to be \NP-complete~\cite{DBLP:conf/coco/Karp72}.
Furthermore, it is $\W[2]$-complete when parameterized by $t'$~\cite{DF99}, and is hard to approximate~\cite{DinurS14}.
Our reduction indeed yields the fixed-parameter intractability and inapproximability. 
\begin{theorem} \label{thm:NPhard}
    Let $(G',t')$ be an instance of \textsc{Dominating Set}, where $G'=(V', E')$.
    Then, there exists an instance $(G, t, k, B)$ of \bvci, where $G = (V,E)$, such that 
        \begin{list}{*}{%
                 \settowidth{\labelwidth}{(a)}%
                   \setlength{\parsep}{0pt}%
                   \setlength{\parskip}{0pt}}
            \item[\textup{(}a\textup{)}] $|V| = 3|V'|$, and $|E|= 2|V'|+2|E'|$\textup{;}
            \item[\textup{(}b\textup{)}] $t = t'$, $k = 2|V'|-1$, and $|B| = |V'|$\textup{;} and 
            \item[\textup{(}c\textup{)}] $G'$ has a dominating set of size at most $t'$ if and only if there exists a solution $X$ to $(G, t, k, B)$.
        \end{list} 
    Furthermore, $(G, t, k, B)$ can be constructed in polynomial time.
\end{theorem}
\begin{proof}
    For a given instance $(G' = (V',E'),t')$ of \textsc{Dominating Set}, we construct such an instance $(G=(V,E), t, k, B)$ of \bvci as follows; 
    see \Cref{fig:reduction} for an example.
    For each vertex $v_i' \in V'$, we create a copy $v_i'$ for $G$.
    We also introduce two new vertices $v_i$ and $\bar{v}_i$ to $G$, and join them by an edge $\{v_i, \bar{v}_i\}$. 
    We then add an edge $\{v_i, v_j'\}$ for every $v_j' \in N_{G'}[\{v_i'\}]$.
    Let $G$ be the resulting graph, where $V = \{ v_i',  v_i, \bar{v}_i \mid v_i' \in V' \}$, and let $B=\{ v_i \mid v_i' \in V' \}$.
    Note that $G$ is a bipartite graph with bipartition $(B, V\setminus B)$. 
    Finally, we set $t = t'$ and $k = |B|+|V'|-1$.
    In this way, $(G, t, k, B)$ can be constructed in polynomial time, and conditions~(a) and~(b) are satisfied.  

\begin{figure}[tbp]
    \centering
    \includegraphics[width=0.75\linewidth]{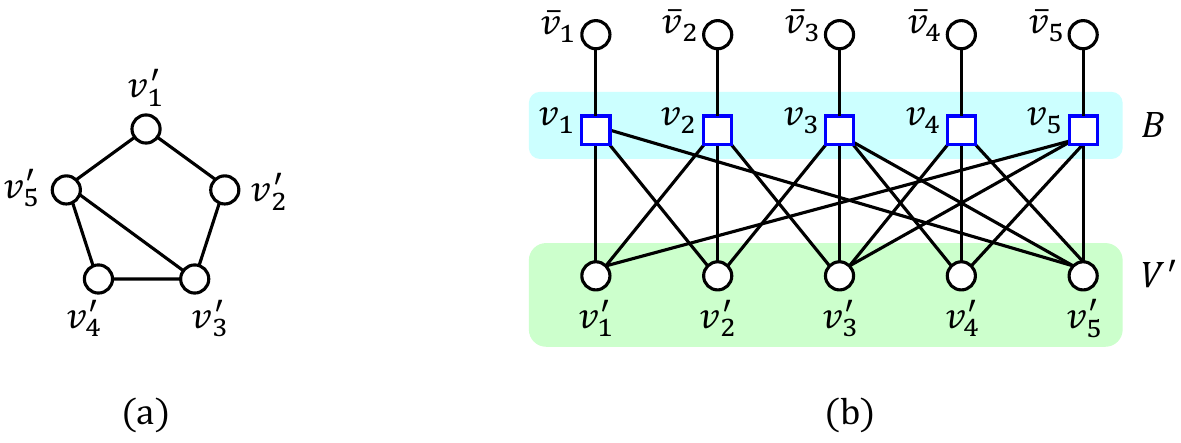} 
    \caption{(a) A graph $G'=(V',E')$ of an instance of \textsc{Dominating Set}, and (b) the corresponding graph $G=(V,E)$ of an instance of \bvci.}
    \label{fig:reduction}
\end{figure}

    It remains to verify condition~(c).
    As we will formalize later, for a solution $X \subseteq B$ to the instance $(G, t, k, B)$, 
    the corresponding set $D' = \{ v'_i \in V' \mid v_i \in X \}$ of $G'$ is a dominating set,
    and vice versa.
    To prove condition~(c), we first establish a necessary and sufficient condition for a vertex cover to be disjoint from $X$ in the following claim. 
    \begin{claim}\label{claim:DCVC}
        For any $X \subseteq B$, it holds that $|N_G(X) \cap V'| \le |V'|-1$ if and only if there exists a vertex cover $C$ of $G$ of size at most $k$ such that $X \cap C = \emptyset$. 
    \end{claim}
    \begin{proof}[Proof of \Cref{claim:DCVC}]
        We first prove the necessity.
        Assume that $|N_G(X) \cap V'| \le |V'|-1$ holds for a subset $X \subseteq B$. 
        Let $C = (B \setminus X) \cup N_G(X)$. 
        Then, we have $X \cap C = \emptyset$.
        Furthermore, since $N_G(X) = \{\bar{v}_i \mid v_i \in X \} \cup (N_G(X) \cap V')$, we have 
        \[
            |C| = (|B|-|X|) + (|X| + |N_G(X) \cap V'|) \le |B| + |V'| - 1 = k.
        \]
        It remains to show that $C$ forms a vertex cover of $G$. 
        Since $B \setminus X \subseteq C$ and $G$ is a bipartite graph with bipartition $(B, V\setminus B)$, it suffices to show that $C \cap e \neq \emptyset$ for every edge $e$ incident to a vertex in $X$.
        This condition is guaranteed by the fact that $N_G(X) \subseteq C$. 

        We then prove the sufficiency.
        Assume that $G$ has a vertex cover $C$ of size at most $k$ such that $X \cap C = \emptyset$.
        Since $C$ is a vertex cover of $G$ and $X \cap C = \emptyset$, we have $N_G(X) \subseteq C$ to cover the edges incident to $X$.
        Recall that $N_G(X)$ is the disjoint union of $\{\bar{v}_i \mid v_i \in X \}$ and $N_G(X) \cap V'$. 
        Furthermore, any vertex cover of $G$ must cover the edge $\{v_j, \bar{v}_j\}$ for each $v_j \in B \setminus X$, and hence $C$ must contain at least one of $v_j$ and $\bar{v}_j$. 
        Since the set $\{\bar{v}_i \mid v_i \in X\}$, the set $V'$, and each set $\{v_j, \bar{v}_j\}$ for $v_j \in B \setminus X$ are pairwise disjoint, we can obtain the following lower bound on the size of $C$:
        \begin{align*}
            |C| &\ge |\{\bar{v}_i \mid v_i \in X\}| + |N_G(X) \cap V'| + \sum_{v_j \in B \setminus X} \left| \{v_j, \bar{v}_j\} \cap C \right| \\
                &\ge |X| + |N_G(X) \cap V'| + |B \setminus X| 
                = |B| + |N_G(X) \cap V'|.
        \end{align*}
        Since we assumed $|C| \le k = |B| + |V'| - 1$, we obtain $|N_G(X) \cap V'| \le |V'| - 1$, as claimed.
    \end{proof}

    We are now ready to verify condition~(c).

    We first prove the necessity.
    Assume that $G'$ has a dominating set $D'$ of size at most $t'$. 
    Let $X = \{ v_i \in B \mid v'_i \in D' \}$.
    Then, we have $|X| = |D'| \le t' = t$.
    Since $D'$ is a dominating set of $G'$, the construction of $G$ implies that $V' \subseteq N_G(X)$.
    Therefore, we have $N_G(X) \cap V' = V'$ and hence $|N_G(X) \cap V'| > |V'| - 1$.
    By the contrapositive of \Cref{claim:DCVC}, there is no vertex cover of size at most $k$ that is disjoint from $X$.
    Therefore, $X$ is a solution to the instance $(G, t, k, B)$.

    We then prove the sufficiency.
    Assume that there exists a solution $X \subseteq B$ to the instance $(G, t, k, B)$.
    Then, $|X| \le t$, and $X \cap C \neq \emptyset$ for every vertex cover $C$ of size at most $k$.
    By \Cref{claim:DCVC}, the latter condition implies that $|N_G(X) \cap V'| > |V'|-1$, and hence we have $N_G(X) \cap V' = V'$.
    For the graph $G'$, let $D' = \{ v'_i \in V' \mid v_i \in X \}$.
    Then, we have $|D'| = |X| \le t = t'$.
    Since $N_G(X) \cap V' = V'$, every vertex $v'_j \in V'$ is adjacent to at least one vertex $v_i \in X$ in $G$.
    By the construction of $G$, the edge $\{v_i, v'_j\}$ in $G$ exists only when $v'_j \in N_{G'}[\{v'_i\}]$, and hence every vertex $v'_j \in V'$ is dominated by some vertex $v'_i \in D'$ in $G'$.
    Therefore, $D'$ is a dominating set of $G'$ of size at most $t'$.
\end{proof}

    Since \textsc{Dominating Set} is \NP-complete~\cite{DBLP:conf/coco/Karp72}, \Cref{prop:np} and \Cref{thm:NPhard} yield the following corollary. 
\begin{corollary}
    \bvci is \NP-complete.
\end{corollary}

    \textsc{Dominating Set} is known to be $\W[2]$-complete when parameterized by $t'$~\cite{DF99}. 
    Since \Cref{thm:NPhard}(b) states that $t=t'$, we have the following corollary. 
\begin{corollary}
    \bvci is $\W[2]$-hard when parameterized by the solution size $t$.
\end{corollary}

    Furthermore, our reduction exactly preserves the optimum solution size (\Cref{thm:NPhard}(b)), and increases the number of vertices by only a constant factor (\Cref{thm:NPhard}(a)).
    Therefore, the inapproximability of \textsc{Dominating Set}~\cite{DinurS14} yields the following corollary.
\begin{corollary}
    Unless $\P = \NP$, there is no polynomial-time approximation algorithm within a factor of $(1 - o(1)) \ln |V|$
    for finding a minimum-size solution $X$ for \bvci.
\end{corollary}

\section{Bipartite Near-Minimum Vertex Cover Interdiction}

In contrast to the \NP-completeness of \bvci, we show that \bminvci can be solved in polynomial time. 
More generally, we show that \bvci can be solved in polynomial time 
when the difference between $k$ and the minimum vertex cover size $\mu(G)$ is bounded by a fixed constant. 

In Section~\ref{sec:dualformulation}, 
we present an alternative formulation of \bvci based on the duality between vertex covers and matchings in bipartite graphs. 
Using this formulation, in Section~\ref{sec:XPalgo}, we give an \XP-time algorithm for \bvci parameterized by 
$k - \mu(G)$, which implies the polynomial-time solvability of \bminvci. 
A faster algorithm for \bminvci is presented in Section~\ref{sec:almostlinear}. 
As a complementary result, in Section~\ref{sec:W1hard},  
we show that \bvci is $\W[1]$-hard when parameterized by $k - \mu(G)$, 
implying that no fixed-parameter tractable algorithm exists under standard complexity assumptions.

\subsection{Dual Formulation}
\label{sec:dualformulation}

Let $G=(V, E)$ be a graph. 
A multiset of $E$ is a collection of edges from $E$ in which edges are allowed to appear multiple times, 
and it can be represented by multiplicities $\alpha_e \in \mathbb{Z}_+$ for each $e \in E$. 
For a multiset $M$ of $E$ and $F \subseteq E$, we denote by $M \cap F$ the multiset 
in which an edge $e \in E$ has the same multiplicity as in $M$ if $e \in F$, and $0$ otherwise. 
For a multiset $M$ of $E$ and $v \in V$, we denote ${\rm deg}_{M}(v) = |M \cap \delta_G(v)|$. 
For $b \in \mathbb{Z}_{+}^V$, a multiset $M$ of $E$ is called a \emph{$b$-matching} if 
${\rm deg}_{M}(v) \le b_v$ for any $v \in V$, where $b_v$ is the entry of $b$ that corresponds to a vertex $v \in V$.

We show the following theorem, which characterizes solutions to \bvci. 

\begin{theorem}
    \label{thm:dualformulate}
    Let $G=(V,E)$ be a bipartite graph and $X\subseteq V$.
    Then, $X$ is a hitting set of the family $\mathcal{C}(G,k)$ of vertex covers of $G$ of size at most $k$ if and only if 
    there exists a multiset\footnote{%
    We can replace ``multiset'' with ``set'' when $X$ contains no pair of adjacent vertices. Indeed, in this case the degree constraint on $V \setminus X$ implies that the multiplicity of each edge is at most one.}
    $M$ of $E$ such that ${\rm deg}_{M}(v) \le 1$ for every $v \in V \setminus X$ and $|M| = k+1$.  
\end{theorem}

\begin{proof}
We define $b \in \mathbb{Z}^{V}_{+}$ by $b_v = 1$ for $v \in V \setminus X$ and $b_v = k+1$ for $v \in X$, 
and consider the following problem: 
    Find a vertex cover $C \subseteq V$ of $G$ that minimizes $\sum_{v \in C} b_v$.    
As in the proof for Proposition~\ref{prop:np}, 
the optimal value of this problem is at least $k+1$ if and only if $X$ is a hitting set of $\mathcal{C}(G,k)$.

This problem can be written as the following integer linear programming problem $\mathsf{IP}(X)$: 
\begin{alignat*}{3}
    \text{minimize}
    &\quad \displaystyle\sum_{v \in V} b_v y_v
    &\quad
    \text{subject to}
    &\quad y_u + y_v \geq 1 
    &&\quad (\forall\{u,v\} \in E), \\
    &\quad
    &\quad
    &\quad y_v \in \mathbb{Z}_{+} 
    &&\quad (\forall v \in V). \nonumber
\end{alignat*}
By the duality between $b$-matchings and vertex covers~ (see~\cite[Theorem 21.1]{Sch03}), 
the optimal value of $\mathsf{IP}(X)$ is equal to the maximum size of a $b$-matching in $G$, that is, 
the optimal value of the following problem $\mathsf{DIP}(X)$:
\begin{alignat*}{3}
    \text{maximize}
    &\quad 
    \displaystyle\sum_{e \in E}\alpha_e
    &\quad
    \text{subject to}
    &\quad \sum_{e \in \delta_G(v)} \alpha_e \leq b_v 
    &&\quad (\forall v\in V), \\
    &\quad
    &\quad
    &\quad \alpha_e \in \mathbb{Z}_+ 
    &&\quad (\forall e \in E). 
\end{alignat*}

Therefore, $X$ is a hitting set of $\mathcal{C}(G,k)$ if and only if there exists a $b$-matching $M$ of size $k+1$. 
Since the constraint ${\rm deg}_{M}(v) \le b_v$ is redundant for $v \in X$, the constraints on $M$ can be rewritten as ${\rm deg}_{M}(v) \le 1$ for every $v \in V \setminus X$ and $|M| = k+1$. 
\end{proof}

We say that such an edge multiset $M$ is a \emph{certificate} for a solution $X \subseteq V$.

\subsection{An XP Algorithm}
\label{sec:XPalgo}

Recall that we denote by $\mu (G)$ the size of a maximum matching of a graph $G$, 
which is equal to the minimum vertex cover size when $G$ is bipartite.
The objective of this subsection is to show the following theorem. 
Note that we refer to \emph{minimal} hitting sets, not minimum ones.

\begin{theorem}
    \label{thm:solbound}
    Let $G=(V, E)$ be a bipartite graph, and $k$ be an integer with $k \ge \mu(G)$.
    Furthermore, let $\ell=k-\mu (G)$. 
    Then, every minimal hitting set $X \subseteq V$ of the family $\mathcal{C}(G,k)$ of the vertex covers of $G$ of size at most $k$ has at most $2(\ell+1)$ vertices.
\end{theorem}

\red{We notice that \Cref{thm:solbound} miserably fails for non-bipartite graphs. For example, let $H$ be a triangle with a degree-one vertex attached to each vertex (so $H$ is a six-vertex graph), and let $G$ be the disjoint union of $k/3$ copies of $H$, where $k$ is a multiple of $3$. Then, $\mu(G) = k$ and so $\ell = k - \mu(G) = 0$. On the other hand, every minimal hitting set of $\mathcal{C}(G,k)$ has $2k/3$ vertices since every minimal hitting set of $\mathcal{C}(H,3)$ has two vertices. Note that in this example the maximum size of a matching is equal to the minimum size of a vertex cover.}

\begin{proof}[Proof \red{of \Cref{thm:solbound}}]
Let $X$ be a minimal hitting set of $\mathcal{C}(G,k)$, and 
let $M^*$ be a maximum matching in $G$. 
Note that $|M^*| = \mu(G)$. 

By \Cref{thm:dualformulate}, there exists a certificate $M$ for a solution $X$, that is, 
$M$ is a multiset of edges in $G$ such that
$|M \cap \delta_G(v)| \le 1$ for every $v \in V \setminus X$ and $|M| = k+1$. 
Among such multisets, we choose $M$ to \red{minimize $|M^* \setminus M|$, 
where $M^* \setminus M$ denotes the set of edges in $M^*$ that do not occur in $M$.}

We now show the following two claims. 

\begin{claim}\label{clm:degMx}
For any $x\in X$, ${\rm deg}_{M}(x) \ge {\rm deg}_{M^*}(x) +1$. 
\end{claim}

\begin{proof}[Proof of \Cref{clm:degMx}]
We see that ${\rm deg}_M(x) \ge 2$, since otherwise $X \setminus \{x\}$ is a hitting set of $\mathcal{C}(G,k)$ with a certificate $M$, contradicting the minimality of $X$.
Furthermore, ${\rm deg}_{M^*}(x) \le 1$ holds as $M^*$ is a matching. 
Therefore, ${\rm deg}_{M}(x) \ge 2 \ge {\rm deg}_{M^*}(x) +1$ as claimed.
\end{proof}

\begin{claim}\label{clm:degMv}
For any $v \in V$, ${\rm deg}_{M}(v) \ge {\rm deg}_{M^*}(v)$. 
\end{claim}

\begin{proof}[Proof of \Cref{clm:degMv}]
Assume to the contrary that there exists a vertex $v \in V$ such that ${\rm deg}_{M}(v) < {\rm deg}_{M^*}(v)$.
Since ${\rm deg}_{M^*}(v) \le 1$, this means that ${\rm deg}_{M^*}(v)=1$ and ${\rm deg}_{M}(v)= 0$.
Let $e = \{v, w\}$ be the unique edge in $M^*$ incident to $v$, where $w \in V \setminus \{v\}$. 

First, consider the case when  ${\rm deg}_{M}(w)= 0$.
Since $|M| > |M^*|$, there exists an edge $e'$ in $M$ 
\red{such that $e' \not\in M^*$ or the multiplicity of $e'$ in $M$ is at least two.} 
Then, $M' := M - e' + e$ is a certificate for $X$, where 
``$-e'$'' (resp.~``$+e$'') denotes decreasing (resp.~increasing) the multiplicity of $e'$ (resp.~$e$) by one. 
This is because ${\rm deg}_{M'}(v)= {\rm deg}_{M'}(w)=1$ and 
${\rm deg}_{M'}(u) \le {\rm deg}_{M}(u)$ for any $u \in V \setminus \{v, w\}$. 
Since \red{$|M^* \setminus M'| = |M^* \setminus M| - 1$}, 
this contradicts the maximality of $|M \cap M^*|$. 

Second, consider the case when ${\rm deg}_{M}(w)\ge 1$.
Let $e'$ be an edge in $M$ incident to $w$. 
Then, $M' := M - e' + e$ is a certificate for $X$, because ${\rm deg}_{M'}(v)=1$ and 
${\rm deg}_{M'}(u) \le {\rm deg}_{M}(u)$ for any $u \in V \setminus \{v\}$. 
Since $M^*$ is a matching and $e \in M^*$ is incident to $w$, $e'$ is not contained in $M^*$, 
which shows that \red{$|M^* \setminus M'| = |M^* \setminus M| - 1$}, contradicting the maximality of $|M \cap M^*|$. 
\end{proof}

By Claims \ref{clm:degMx} and \ref{clm:degMv}, we obtain
\begin{equation*}
2|M| = \sum_{v \in V} {\rm deg}_M(v) \ge \sum_{v \in V} {\rm deg}_{M^*}(v) + |X| = 2|M^*| + |X|. 
\end{equation*}
Therefore, $|X| \le 2|M| - 2|M^*| = 2 (k+1) - 2 \mu(G) = 2 (\ell + 1)$. 
\end{proof}

\begin{corollary}\label{cor:XPinl}
\bvci\ can be solved in $|V|^{O(\ell)}$ time, where $\ell=k-\mu (G)$. 
\end{corollary}

\begin{proof}
Suppose we are given an instance $(G, t, k, B)$ of \bvci. 

If $\ell=k-\mu (G) < 0$, then we can immediately conclude that 
$X = \emptyset$ is a solution. Here we note that 
$\mu(G)$ can be computed in $|V|^{O(1)}$ time. 

When $\ell \ge 0$, 
we consider the following algorithm. 
We enumerate all the sets $X \subseteq B$ with $|X| \le \min \{2(\ell +1), t\}$. 
Then, there are at most $|V|^{2\ell + 2}$ choices of $X$. 
For each $X$, we test whether $X$ is a solution to \bvci\ or not, 
which can be done in $|V|^{O(1)}$ time by finding a minimum weight vertex cover in $G$ (see the proof for Proposition \ref{prop:np}). 
If no solution is found during this procedure, we conclude that the instance $(G, t, k, B)$ has no solution. 
This algorithm runs in $|V|^{O(\ell)}$ time, and 
the validity is guaranteed by \Cref{thm:solbound}. 
\end{proof}

Since \bminvci\ amounts to the case of $\ell=0$, Corollary~\ref{cor:XPinl} implies that it can be solved in polynomial time. 
Notice that the next corollary can also be obtained from a result by
Bentz, Costa, Picouleau, Ries, and de Werra~\cite[Theorem 5.4]{DBLP:journals/jda/BentzCPRW12}.

\begin{corollary}\label{cor:polyminimum}
\bminvci\ can be solved in polynomial time. 
\end{corollary}

In contrast, we note again that this problem becomes $\Sigma_2^p$-complete when $G$ is allowed to be a non-bipartite graph, as we will see in Appendix~\ref{sec:sigma2pmin}.

\subsection{A Faster Algorithm for \bminvci}
\label{sec:almostlinear}

As shown in Corollary \ref{cor:polyminimum}, \bminvci\ can be solved in polynomial time. 
In this subsection, we present a faster algorithm for this problem by revealing the structure of instances with small solutions. 

By \Cref{thm:solbound} with $\ell = 0$,  
if an instance of \bminvci\ has a solution, then there exists a solution of size at most $2$. 
Let $\minvcfamily$ denote the family of minimum vertex covers of $G$. 
Since $\minvcfamily \neq \emptyset$, 
we can see that there exists no solution of size $0$ (i.e., $X=\emptyset$ is not a solution).  
To characterize instances that have a solution of size $1$ or $2$, 
we introduce an auxiliary digraph.

Let $G=(V, E)$ be a bipartite graph with bipartition $(V_1, V_2)$ of $V$. 
For a maximum matching $M^* \subseteq E$ in $G$, 
define an auxiliary digraph $D(M^*) = (V, A)$  
by
\begin{align*}
A_1 &= \{ (v_1, v_2) \mid  v_1 \in V_1,\ v_2 \in V_2,\ \{v_1, v_2\} \in E \}, \\
A_2 &=  \{ (v_2, v_1) \mid  v_1 \in V_1,\ v_2 \in V_2,\ \{v_1, v_2\} \in M^* \}, \\
A   &= A_1 \cup A_2. 
\end{align*}
That is, $A_1$ is obtained from $E$ by orienting every edge from $V_1$ to $V_2$, 
$A_2$ is obtained from $M^*$ by orienting every edge from $V_2$ to $V_1$, 
and $A$ is their union. 
For $i \in \{1, 2\}$, let $U_i$ be the set of vertices in $V_i$ not covered by $M^*$. 
For a dipath $P$ in $D(M^*)$, let $E(P) \subseteq E$ denote the set of edges corresponding to $P$. 
Note that $E(P)$ forms a path that traverses edges in $E \setminus M^*$ and $M^*$ alternately. 
Note also that $D(M^*)$ has no dipath $P$ from $U_1$ to $U_2$, since otherwise $M^* \mathbin{\triangle} E(P)$ is a larger matching than $M^*$.

By using the auxiliary digraph, the minimum solution size can be characterized as follows.

\begin{lemma}\label{lem:size1}
Let $G=(V, E)$ be a bipartite graph, 
let $B \subseteq V$, and let $M^*$ be a maximum matching in $G$. 
Then, there exists a hitting set $X \subseteq B$ of $\minvcfamily$ such that $|X| = 1$
if and only if 
$D(M^*)$ contains a dipath from $U_1$ to $V_2 \cap B$ or from $V_1 \cap B$ to $U_2$. 
\end{lemma}

\begin{proof}
We first show the sufficiency (``if'' part). 
Suppose that $D(M^*)$ contains a dipath $P$ from $u \in U_1$ to $x \in V_2 \cap B$. 
Then, $M := M^* \mathbin{\triangle} E(P)$ is an edge set such that 
$|M| = |M^*| + 1$, ${\rm deg}_{M}(v) = {\rm deg}_{M^*}(v) \le 1$ for $v \in V \setminus \{u, x\}$, and ${\rm deg}_{M}(u) = 1$. 
This shows that $X = \{x\}$ is a hitting set of $\minvcfamily$, 
because $M$ satisfies the conditions in \Cref{thm:dualformulate} with $k = \mu(G)$. 
A similar argument applies when $D(M^*)$ contains a dipath from $V_1 \cap B$ to $U_2$. 

We next show the necessity (``only if'' part). 
Suppose that there exists a hitting set $X = \{x\}$ of $\minvcfamily$,
where $x \in V_1 \cap B$; 
a similar argument applies when $x \in V_2 \cap B$. 
Note that $X$ is a minimal hitting set, because 
$\emptyset$ is not a hitting set. 
By \Cref{thm:dualformulate} with $k = \mu(G)$, there exists a multiset $M$ of edges in $E$ such that $|M| = |M^*|+1$ and 
$|M \cap \delta_G(v)| \le 1$ for every $v \in V - x$. 
Among such multisets, we choose $M$ to 
\red{minimize $|M^* \setminus M|$.}
Then, as in the proof of Theorem~\ref{thm:solbound}, 
Claims \ref{clm:degMx} and \ref{clm:degMv} hold. 

Let $A(M)$ be the multiset of arcs in $A_1$ that is obtained from $M$ by orienting every edge from $V_1$ to $V_2$, and  
let $A' = A(M) \cup A_2$. 
Note that $A'$ is a multiset of arcs in $A$. 
Since $x \in V_1$, the outdegree and indegree of $x$ in $A'$ are ${\rm deg}_{M}(x)$ and ${\rm deg}_{M^*}(x)$, respectively. 
Thus, Claim \ref{clm:degMx} implies that $x$ has more outgoing arcs than incoming arcs in $A'$. 
Then, by the standard flow decomposition argument (see~\cite[Section 10.3]{Sch03}), 
$A'$ contains a dipath $P$ from $x$ to some vertex $y \in V$ such that 
$y$ has more incoming arcs than outgoing arcs in $A'$.

If $y \in V_1$, then the outdegree and indegree of $y$ in $A'$ are ${\rm deg}_{M}(y)$ and ${\rm deg}_{M^*}(y)$, respectively. 
In this case, we obtain ${\rm deg}_{M}(y) < {\rm deg}_{M^*}(y)$, a contradiction to Claim~\ref{clm:degMv}.
Therefore, $y \in V_2$ holds. 
Then, the outdegree and indegree of $y$ in $A'$ are ${\rm deg}_{M^*}(y)$ and ${\rm deg}_{M}(y)$, respectively,  
and ${\rm deg}_{M}(y) > {\rm deg}_{M^*}(y)$ holds. 
Since $y\neq x$ implies $y \in V \setminus X$, 
we obtain ${\rm deg}_{M}(y) \le 1$, and hence ${\rm deg}_{M^*}(y)=0$. 
Therefore, $y \in U_2$ holds and $P$ is a desired dipath.
\end{proof}

\begin{lemma}\label{lem:size2}
Let $G=(V, E)$ be a bipartite graph, 
let $B \subseteq V$, and let $M^*$ be a maximum matching in $G$.
Suppose that there exists no hitting set $X' \subseteq B$ of $\minvcfamily$ such that $|X'| = 1$.
Then, there exists a hitting set $X \subseteq B$ of $\minvcfamily$
such that $|X| = 2$
if and only if 
$D(M^*)$ contains a dipath from $V_1 \cap B$ to $V_2 \cap B$. 
\end{lemma}

\begin{proof}
We first show the sufficiency (``if'' part). 
Suppose that $D(M^*)$ contains a dipath $P$ from $x_1 \in V_1 \cap B$ to $x_2 \in V_2 \cap B$. 
If $E(P)$ consists of a single edge $e$, then define $M$ as the multiset of $E$ obtained from $M^*$ by increasing the multiplicity of $e$ by one. 
Otherwise, let $M := M^* \mathbin{\triangle} E(P)$. 
Then, $M$ is an edge (multi)set such that 
$|M| = |M^*| + 1$, and ${\rm deg}_{M}(v) = {\rm deg}_{M^*}(v) \le 1$ for $v \in V \setminus \{x_1, x_2\}$. 
This shows that $X = \{x_1, x_2\}$ is a hitting set of $\minvcfamily$,
because $M$ satisfies the conditions in \Cref{thm:dualformulate} with $k = \mu(G)$.

We next show the necessity (``only if'' part). 
Suppose that there exists a hitting set $X$ of $\minvcfamily$ such that $|X| =2$. 
Suppose also that $X$ contains a vertex $x$ in $V_1 \cap B$; 
a similar argument applies when $X$ contains a vertex in $V_2 \cap B$. 
By \Cref{thm:dualformulate} with $k = \mu(G)$, there exists a multiset $M$ of edges in $E$ such that $|M| = |M^*|+1$ and 
$|M \cap \delta_G(v)| \le 1$ for every 
\red{$v \in V \setminus X$}. 
Among such multisets, we choose $M$ to 
\red{minimize $|M^* \setminus M|$.}
Then, Claims \ref{clm:degMx} and \ref{clm:degMv} hold, 
because $X$ is a minimal hitting set. 

Let $A(M)$ be the multiset of arcs in $A_1$ that is obtained from $M$ by orienting every edge from $V_1$ to $V_2$, and 
let $A'$ be the union of $A(M)$ and $A_2$. 
By the same argument as in the proof for Lemma~\ref{lem:size1}, 
$A'$ contains a dipath $P$ from $x$ to some vertex $y \in V_2$ such that 
${\rm deg}_{M}(y) > {\rm deg}_{M^*}(y)$.
Since the condition in Lemma \ref{lem:size1} is not satisfied, $y$ is not contained in $U_2$, that is, ${\rm deg}_{M^*}(y)\neq 0$. 
Therefore, ${\rm deg}_{M}(y) > {\rm deg}_{M^*}(y) \ge 1$, and hence $y \in X$. 
This shows that $y \in V_2 \cap B$ holds and $P$ is a desired dipath.
\end{proof}

Given an instance of \bminvci, we can compute a maximum matching $M^*$ in $G$ in $|E|^{1 + o(1)}$ time~\cite{Chen25}
(see also~\cite{ChuzhoyK24,HopcroftK73} for other fast combinatorial algorithms), and then 
construct $D(M^*)$ in linear time. 
By checking the conditions in Lemmas~\ref{lem:size1} and~\ref{lem:size2}, 
we can test whether the instance has a solution of size $1$ or $2$ in linear time. 
If there is no solution of size at most $2$, then 
we can conclude that \bminvci has no solution by \Cref{thm:solbound}. 
Therefore, we obtain the following theorem. 

\begin{theorem}
\bminvci can be solved in $|E|^{1 + o(1)}$ time. 
\end{theorem}

\subsection{$\W[1]$-Hardness}
\label{sec:W1hard}

As a complementary result to Corollary~\ref{cor:XPinl}, we prove the following theorem, 
implying that \bvci is unlikely to admit an FPT algorithm parameterized by $\ell$.

\begin{theorem}
\bvci\ is $\W[1]$-hard when parameterized by $\ell = k - \mu(G)$. 
\end{theorem}

\begin{proof}
We reduce {\sc Maximum Clique} to \bvci. 
In {\sc Maximum Clique}, we are given a simple graph $G=(V, E)$ and a positive integer $p$, 
and the objective is to determine whether $G$ contains a clique of size $p$ (i.e., a complete subgraph with $p$ vertices). 
This problem is known to be $\W[1]$-hard when parameterized by $p$ (see e.g.,~\cite{DF99}). 

\begin{figure}[tbp]
    \centering
    \includegraphics[width=0.8\linewidth]{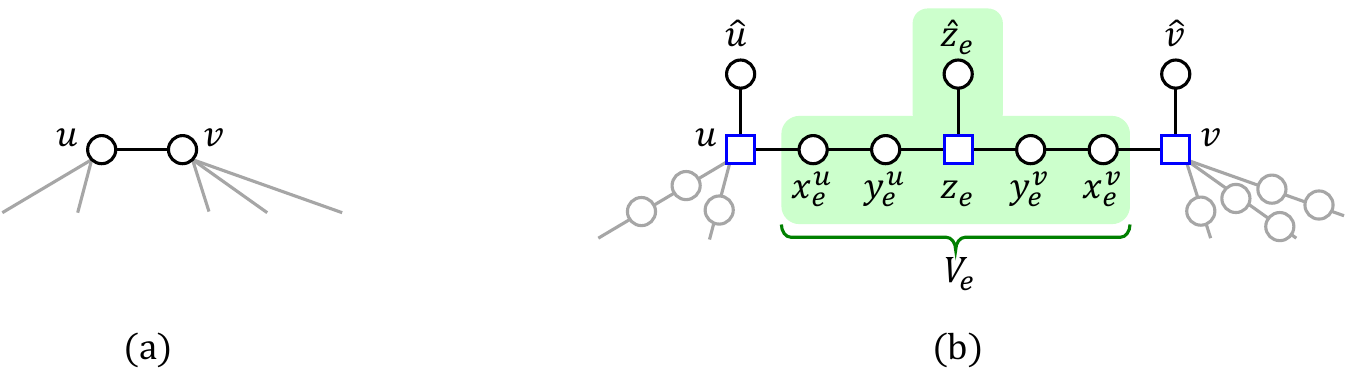} 
    \caption{(a) An edge $\{u, v\}$ in a graph $G$ of \textsc{Maximum Clique}, and (b) the corresponding part in the graph $H$ of \bvci, together with $\hat{u}$ and $\hat{v}$ for $u$ and $v$, respectively. Squares represent vertices in $B$.}
    \label{fig:reduction_W}
\end{figure}

Suppose we are given a graph $G=(V, E)$ and a positive integer $p$ that form an instance of {\sc Maximum Clique}. 
It suffices to give a reduction for the case of $p \ge 4$. 
We construct an instance of \bvci\ as follows (see also \Cref{fig:reduction_W}).
First, replace each edge $e= \{u, v\}$ in $G$ with a path $P_e$ with six edges between $u$ and $v$, and 
let $x^{u}_e, y^{u}_e, z_e, y^{v}_e, x^{v}_e$ be the internal vertices of this path, 
in the order they appear along the path from $u$ to $v$.  
Let $Z = \{z_e \mid e \in E\}$. 
Second, for each vertex $w \in V \cup Z$, 
introduce a new vertex $\hat w$ and a new edge $\{w, \hat w\}$. 
Let $H = (V', E')$ be the obtained graph, that is, 
\[
V' = V \cup \hat V \cup \bigcup_{e \in E} V_e, 
\quad
E' = \{\{w, \hat w\} \mid w \in V \cup Z\} \cup \bigcup_{e \in E} E(P_e),
\]
where $V_e = \{x^{u}_e, y^{u}_e, z_e, y^{v}_e, x^{v}_e, \hat{z_e} \}$ for each $e = \{u, v\} \in E$ and $\hat V = \{\hat w \mid w \in V\}$. 
One can easily see that $H$ is a bipartite graph with $|V'|=2|V|+6|E|$ vertices, which has a perfect matching. 
Therefore, the minimum vertex cover size $\mu(H)$ is equal to $\frac{|V'|}{2} = |V| + 3 |E|$. 
Let 
\[
B = V \cup Z, \quad
t = p+{p \choose 2}, \quad
k = |V| + 3 |E| + p^2 - p - 1.
\]
Note that $k - \mu(H) = p^2 - p - 1$, which is a function of the original parameter $p$. 
Therefore, to show the $\W[1]$-hardness of \bvci, 
it suffices to prove that the obtained instance $(H, t, k, B)$ of \bvci\ is equivalent to the original instance $(G, p)$ of {\sc Maximum Clique}. 
This equivalence is established by the following claims.

\begin{claim}
If $G$ contains a clique of size $p$, then there exists a hitting set $X \subseteq B$ for $\mathcal{C}(H, k)$ such that $|X| \le t$. 
\end{claim}

\begin{proof}
Let $U \subseteq V$ and $F \subseteq E$ be the vertex set and the edge set of a clique of size $p$ in $G$, respectively. 
We prove that $X := U \cup \{z_e \mid e \in F\} \subseteq B$ is a desired set. 
Since $|U| = p$ and $|F| = {p \choose 2}$, it holds that $|X| = t$. 

To prove that $X$ is a hitting set for $\mathcal{C}(H, k)$, 
we show that every vertex cover of $H$ that is disjoint from $X$ has size greater than $k$. 
Let $C$ be a vertex cover of $H$ such that $C \cap X = \emptyset$. 
Then, $C$ satisfies the following properties.  
\begin{itemize}
\item
For each $v \in V$, to cover the edge $\{v, \hat v\}$, $C$ contains at least one of $v$ and $\hat v$. 
\item
For each $e = \{u, v\} \in F$, $C \cap \{u, v, z_e\} = \emptyset$ holds, because $u, v \in U \subseteq X$ and $z_e \in X$. 
Thus,  
in order to cover the edges in $E(P_e) \cup \{ \{z_e, \hat {z_e}\} \}$, 
$C$ has to contain all the vertices in $V_e \setminus \{z_e\}$, implying that $|C \cap V_e| = 5$. 
\item
For each $e = \{u, v\} \in E \setminus F$, $C \cap V_e$ covers the edges 
$\{x^u_e, y^u_e\}, \{y^v_e, x^v_e\}$, and $\{z_e, \hat {z_e}\}$. 
Thus, 
$C$ contains at least three vertices in $V_e$. 
\end{itemize}
Therefore, 
\begin{align*}
|C| &= \sum_{v \in V} | C \cap \{v, \hat v\}| + \sum_{e \in F} |C \cap V_e| + \sum_{e \in E \setminus F} |C \cap V_e| \\ 
     &\ge |V| + 5|F| + 3(|E| - |F|) 
      = |V| + 3|E| + p^2 - p 
      >  k. 
\end{align*}
This shows that 
$X$ is a hitting set for $\mathcal{C}(H, k)$, which completes the proof. 
\end{proof}

\begin{claim}
If there exists a hitting set $X \subseteq B$ for $\mathcal{C}(H, k)$ such that 
$|X| \le t$, 
then $G$ contains a clique of size $p$.  
\end{claim}

\begin{proof}
Let $X \subseteq B$ be a hitting set for $\mathcal{C}(H, k)$ such that $|X| \le t$. 
Let $U = X \cap V$ and $F = \{e \in E \mid z_e \in X\}$. 
Our aim is to show that $U$ and $F$ are the vertex set and the edge set of a clique of size $p$ in $G$, respectively. 

Let $\mathcal{P}$ be the set of all pairs $(v, e)$ consisting of a vertex $v \in V$ and an edge $e \in E$ such that 
$e$ is incident to $v$ in $G$; that is, $\mathcal{P} = \{(v, e) \mid v \in V,\ e \in \delta_G(v)\}$.  
For each $(v, e) \in \mathcal{P}$, define $C_{(v, e)} \subseteq V'$ as 
\[
C_{(v, e)} = 
\begin{cases}
\{x^v_e\} & \mbox{if $e \in E \setminus F$,} \\
\{y^v_e\} & \mbox{if $v \in V \setminus U$ and $e \in F$,} \\
\{x^v_e, y^v_e\} & \mbox{if $v \in U$ and $e \in F$.} 
\end{cases}
\]
Then, define a vertex set $C \subseteq V'$ of $H$ as 
\[
C = \{\hat  w \mid w \in X\} \cup (B \setminus X) \cup \bigcup_{(v, e) \in \mathcal{P}} C_{(v, e)}. 
\]
One can see that $C$ is a vertex cover of $H$ such that $C \cap X = \emptyset$ and 
\begin{align}
|C| &= |X| + (|B| - |X|) + |\mathcal{P}| + |\{(v, e) \in \mathcal{P} \mid v \in U,\ e \in F\}| \notag \\
    &= |V| + 3|E| + |\{(v, e) \in \mathcal{P} \mid v \in U,\ e \in F\}|, \label{eq:w1hard01} 
\end{align} 
where we note that $|B| = |V| + |E|$ and $|\mathcal{P}| = 2|E|$. 
Since $X$ is a hitting set for $\mathcal{C}(H, k)$, we obtain
\begin{equation}
|C| \ge k+1 = |V| + 3|E| + p^2 - p. \label{eq:w1hard02}  
\end{equation}
Hence,  
\begin{equation}
|\{(v, e) \in \mathcal{P} \mid v \in U,\ e \in F\}| \ge p^2 - p \label{eq:w1hard03}
\end{equation}
by (\ref{eq:w1hard01}) and (\ref{eq:w1hard02}). 
This implies that $|U| \ge 1$. 

Observe that each edge in $F$ is contained in at most two pairs in 
$\{(v, e) \in \mathcal{P} \mid v \in U,\ e \in F\}$, and so 
$|\{(v, e) \in \mathcal{P} \mid v \in U,\ e \in F\}| \le 2 |F|$. 
This, together with (\ref{eq:w1hard03}), shows that 
$|F| \ge \frac{p^2 - p}{2}$. Therefore, we obtain 
\begin{equation}
|U| = |X| - |F| \le t - \frac{p^2 - p}{2} = p. \label{eq:w1hard06}
\end{equation}

Let $F_1$ (resp.~$F_2$) be the 
set of edges $e \in F$ such that exactly one (resp.~two) end vertices of $e$ belong to $U$. 
Since $|F_2| \le {|U| \choose 2}$ holds, we obtain
\begin{equation}
|\{(v, e) \in \mathcal{P} \mid v \in U,\ e \in F\}| = |F_1| + 2|F_2| \le |F| + |F_2| \le |F| + {|U| \choose 2}. \label{eq:w1hard04}
\end{equation} 
By (\ref{eq:w1hard03}) and (\ref{eq:w1hard04}), it holds that
\begin{equation}
p^2 - p
\le |F| + {|U| \choose 2}
\le \left( p + {p \choose 2} - |U|\right) + {|U| \choose 2} 
\label{eq:w1hard05}
\end{equation} 
since $|F| = |X|-|U| \leq t - |U| \leq p + \binom{p}{2} - |U|$.
Rearranging this inequality, we obtain $\frac{1}{2} (|U|^2 - 3|U|) \ge \frac{1}{2} (p^2 -3p)$, which is equivalent to
$\left( |U| - \frac{3}{2} \right)^2 \ge \left( p - \frac{3}{2} \right)^2$. 
Since we have assumed that $p \ge 4$, this implies that $|U| \ge p$.

This, together with (\ref{eq:w1hard06}), shows that $|U| = p$ and all of the above inequalities hold with equality.
Equality in (\ref{eq:w1hard05}) implies $|F| = \binom{p}{2}$, and by equality in (\ref{eq:w1hard04}), 
every edge in $F$ has both of its end vertices in $U$. 
Therefore, $U$ and $F$ form a clique of size $p$ in $G$.
\end{proof}

These claims show the validity of the reduction, and hence \bvci\ is $\W[1]$-hard. 
\end{proof}

\section{Parameterized by Vertex Cover Size}

In this section, we show that \bvci is fixed-parameter tractable when parameterized by the upper bound $k$ on the vertex cover size. 
More strongly, we show that \bvci admits a polynomial kernel.

\begin{theorem}
    We can construct a kernel \red{with $O(k^2)$ vertices} for \bvci in polynomial time. 
\end{theorem}

\begin{proof}
Suppose that we are given an instance $(G, t, k, B)$ of \bvci. 
We see that there exists a solution of size zero (i.e., $X = \emptyset$ is a solution) 
if and only if $\mu(G) > k$. 
Since this condition can be checked in polynomial time, 
we may assume that $t \ge 1$ and $\mu(G) \le k$. 

We consider the following two reduction procedures.
(R1) If there exists $v \in B$ \red{of} degree at least $k+1$, then return $X = \{v\}$ as a solution of size one.
(R2) If there exists $v \in V \setminus B$ \red{of} degree at least $k+1$, then remove $v$ from $G$ and decrease $k$ by one; that is, consider $(G - v, t, k - 1, B)$.

If there exists $v \in V$ \red{of} degree at least $k+1$, then every vertex cover of size at most $k$ must contain $v$, which shows the validity of 
the reduction procedures above. 
Note that if we apply (R2), then $\mu(G-v) = \mu(G) - 1 \le k-1$ holds; that is, the assumption still holds. 

By applying these procedures exhaustively, we obtain either a solution of size one or an equivalent instance $(G', t, k', B)$ such that $k' \le k$, $\mu(G') \le k'$, and every vertex $v$ in $G'$ has degree at most $k'$.
Since each edge in $G'$ is incident to a vertex in a minimum vertex cover, we have $|E(G')| \le (k')^2 \le k^2$.
We then remove all isolated vertices, which obviously does not affect feasibility. 
Since the resulting graph has at most $2|E(G')| \le 2k^2$ vertices, 
we obtain a kernel \red{with $O(k^2)$ vertices}.
\end{proof}

This theorem implies the fixed-parameter tractability of \bvci as follows. 

\begin{corollary}\label{cor:fptk}
    \bvci\ is fixed-parameter tractable when parameterized by $k$.
\end{corollary}

\red{We can also show the fixed-parameter tractability parameterized by $|V \setminus B|$ as follows.} 

\begin{corollary}
    \red{\bvci\ is fixed-parameter tractable when parameterized by $|V \setminus B|$.}
\end{corollary}

\begin{proof}
If $t \le 1$, then 
for every $X \subseteq B$ with $|X| \le 1$, 
we test whether $X$ is a solution to the given instance of \bvci\ or not, 
which can be done in $|V|^{O(1)}$ time as in the proof for Proposition \ref{prop:np}. 
Therefore, it suffices to consider the case where $t \ge 2$.

If $V \setminus B$ is not a vertex cover of $G$, then $B$ contains a pair of adjacent vertices $u$ and $v$,
which means that $X = \{u,v\}$ is a solution of size two.

If $V \setminus B$ is a vertex cover of size at most $k$, 
then no set $X \subseteq B$ hits $V \setminus B$, 
and hence we can conclude that the instance of $\bvci$ has no solution.

We are left with the case where $|V \setminus B| \ge k+1$.
In this case,  Corollary~\ref{cor:fptk} implies that \bvci is fixed-parameter tractable when parameterized by $|V \setminus B|$.
\end{proof}

Note that \vci, without bipartiteness, is \coNP-complete when $B = V$ \cite{DBLP:journals/tcs/GruneW26}.


\bibliography{biblio}

@inproceedings{DBLP:conf/ipco/GruneW25,
  author       = {Christoph Gr{\"{u}}ne and
                  Lasse Wulf},
  editor       = {Nicole Megow and
                  Amitabh Basu},
  title        = {Completeness in the Polynomial Hierarchy for Many Natural Problems
                  in Bilevel and Robust Optimization},
  booktitle    = {Integer Programming and Combinatorial Optimization - 26th International
                  Conference, {IPCO} 2025, Baltimore, MD, USA, June 11--13, 2025, Proceedings},
  series       = {Lecture Notes in Computer Science},
  volume       = {15620},
  pages        = {256--269},
  publisher    = {Springer},
  year         = {2025},
  doi          = {10.1007/978-3-031-93112-3_19}
}

@article{DBLP:journals/dam/BazganTT11,
  author       = {Cristina Bazgan and
                  Sonia Toubaline and
                  Zsolt Tuza},
  title        = {The most vital nodes with respect to independent set and vertex cover},
  journal      = {Discret. Appl. Math.},
  volume       = {159},
  number       = {17},
  pages        = {1933--1946},
  year         = {2011},
  doi          = {10.1016/J.DAM.2011.06.023}
}

@inproceedings{DBLP:conf/coco/Karp72,
  author       = {Richard M. Karp},
  editor       = {Raymond E. Miller and
                  James W. Thatcher},
  title        = {Reducibility Among Combinatorial Problems},
  booktitle    = {Proceedings of a Symposium on the Complexity of Computer Computations,
                  held March 20--22, 1972, at the {IBM} Thomas J. Watson Research Center,
                  Yorktown Heights, New York, {USA}},
  series       = {The {IBM} Research Symposia Series},
  pages        = {85--103},
  publisher    = {Plenum Press, New York},
  year         = {1972},
  doi          = {10.1007/978-1-4684-2001-2_9}
}

@book{DF99,
  author    = {Rodney G. Downey and Michael R. Fellows},
  title     = {Parameterized Complexity},
  publisher = {Springer},
  year      = {1999},
  series    = {Monographs in Computer Science},
  address   = {New York}
}

@book{Sch03,
  title={Combinatorial Optimization: Polyhedra and Efficiency},
  author={Schrijver, Alexander},
  year={2003},
  publisher={Springer}
}

@article{Chen25,
  author       = {Li Chen and
                  Rasmus Kyng and
                  Yang P. Liu and
                  Richard Peng and
                  Maximilian Probst Gutenberg and
                  Sushant Sachdeva},
  title        = {Maximum Flow and Minimum-Cost Flow in Almost-Linear Time},
  journal      = {J. {ACM}},
  volume       = {72},
  number       = {3},
  pages        = {19:1--19:103},
  year         = {2025},
  url          = {https://doi.org/10.1145/3728631},
  doi          = {10.1145/3728631},
}

@inproceedings{ChuzhoyK24,
  author       = {Julia Chuzhoy and
                  Sanjeev Khanna},
  editor       = {David P. Woodruff},
  title        = {A Faster Combinatorial Algorithm for Maximum Bipartite Matching},
  booktitle    = {Proceedings of the 2024 {ACM-SIAM} Symposium on Discrete Algorithms,
                  {SODA} 2024, Alexandria, VA, USA, January 7--10, 2024},
  pages        = {2185--2235},
  publisher    = {{SIAM}},
  year         = {2024},
  url          = {https://doi.org/10.1137/1.9781611977912.79},
  doi          = {10.1137/1.9781611977912.79},
}

@article{HopcroftK73,
  author       = {John E. Hopcroft and
                  Richard M. Karp},
  title        = {An n\({}^{\mbox{5/2}}\) Algorithm for Maximum Matchings in Bipartite
                  Graphs},
  journal      = {{SIAM} J. Comput.},
  volume       = {2},
  number       = {4},
  pages        = {225--231},
  year         = {1973},
  url          = {https://doi.org/10.1137/0202019},
  doi          = {10.1137/0202019},
}

@inproceedings{DinurS14,
  author       = {Irit Dinur and
                  David Steurer},
  editor       = {David B. Shmoys},
  title        = {Analytical approach to parallel repetition},
  booktitle    = {Symposium on Theory of Computing, {STOC} 2014, New York, NY, USA,
                  May 31--June 03, 2014},
  pages        = {624--633},
  publisher    = {{ACM}},
  year         = {2014},
  url          = {https://doi.org/10.1145/2591796.2591884},
  doi          = {10.1145/2591796.2591884},
  bibsource    = {dblp computer science bibliography, https://dblp.org}
}

@article{DBLP:journals/eor/SmithS20,
  author       = {J. Cole Smith and
                  Yongjia Song},
  title        = {A survey of network interdiction models and algorithms},
  journal      = {Eur. J. Oper. Res.},
  volume       = {283},
  number       = {3},
  pages        = {797--811},
  year         = {2020},
  doi          = {10.1016/J.EJOR.2019.06.024}
}

@article{DBLP:journals/mor/ChestnutZ17,
  author       = {Stephen R. Chestnut and
                  Rico Zenklusen},
  title        = {Interdicting Structured Combinatorial Optimization Problems with $\{0,
                  1\}$-Objectives},
  journal      = {Math. Oper. Res.},
  volume       = {42},
  number       = {1},
  pages        = {144--166},
  year         = {2017},
  doi          = {10.1287/MOOR.2016.0798}
}

@article{DBLP:journals/csr/AusielloBFLR26,
  author       = {Giorgio Ausiello and
                  Lorenzo Balzotti and
                  Paolo Giulio Franciosa and
                  Isabella Lari and
                  Andrea Ribichini},
  title        = {Interdiction in network maximum flow and related problems: {A} survey},
  journal      = {Comput. Sci. Rev.},
  volume       = {60},
  pages        = {100867},
  year         = {2026},
  doi          = {10.1016/J.COSREV.2025.100867}
}

@inproceedings{DBLP:conf/icalp/ChenWZ22,
  author       = {Lin Chen and
                  Xiaoyu Wu and
                  Guochuan Zhang},
  editor       = {Mikolaj Bojanczyk and
                  Emanuela Merelli and
                  David P. Woodruff},
  title        = {Approximation Algorithms for Interdiction Problem with Packing Constraints},
  booktitle    = {49th International Colloquium on Automata, Languages, and Programming,
                  {ICALP} 2022, Paris, France, July 4--8, 2022},
  series       = {LIPIcs},
  pages        = {39:1--39:19},
  publisher    = {Schloss Dagstuhl - Leibniz-Zentrum f{\"{u}}r Informatik},
  year         = {2022},
  doi          = {10.4230/LIPICS.ICALP.2022.39}
}

@article{DBLP:journals/algorithmica/DvorakK18,
  author       = {Pavel Dvor{\'{a}}k and
                  Dusan Knop},
  title        = {Parameterized Complexity of Length-bounded Cuts and Multicuts},
  journal      = {Algorithmica},
  volume       = {80},
  number       = {12},
  pages        = {3597--3617},
  year         = {2018},
  doi          = {10.1007/S00453-018-0408-7}
}

@article{DBLP:journals/networks/BazganFNNS19,
  author       = {Cristina Bazgan and
                  Till Fluschnik and
                  Andr{\'{e}} Nichterlein and
                  Rolf Niedermeier and
                  Maximilian Stahlberg},
  title        = {A more fine-grained complexity analysis of finding the most vital
                  edges for undirected shortest paths},
  journal      = {Networks},
  volume       = {73},
  number       = {1},
  pages        = {23--37},
  year         = {2019},
  doi          = {10.1002/NET.21832}
}

@article{DBLP:journals/tcs/GruneW26,
  author       = {Christoph Gr{\"{u}}ne and
                  Lasse Wulf},
  title        = {The complexity of blocking all solutions},
  journal      = {Theor. Comput. Sci.},
  volume       = {1069},
  pages        = {115820},
  year         = {2026},
  doi          = {10.1016/J.TCS.2026.115820}
}

@article{DBLP:journals/jda/BentzCPRW12,
  author       = {C{\'{e}}dric Bentz and
                  Marie{-}Christine Costa and
                  Christophe Picouleau and
                  Bernard Ries and
                  Dominique de Werra},
  title        = {$d$-{T}ransversals of stable sets and vertex covers in weighted bipartite
                  graphs},
  journal      = {J. Discrete Algorithms},
  volume       = {17},
  pages        = {95--102},
  year         = {2012},
  doi          = {10.1016/J.JDA.2012.06.002}
}

@article{DBLP:journals/jco/CostaWP11,
  author       = {Marie{-}Christine Costa and
                  Dominique de Werra and
                  Christophe Picouleau},
  title        = {Minimum $d$-blockers and $d$-transversals in graphs},
  journal      = {J. Comb. Optim.},
  volume       = {22},
  number       = {4},
  pages        = {857--872},
  year         = {2011},
  doi          = {10.1007/S10878-010-9334-6}
}

\appendix

\section{$\Sigma_2^p$-Completeness for the Minimum Vertex Cover Interdiction Problem in General Graphs}
\label{sec:sigma2pmin}

In this appendix, we study the following problem.

\begin{tcolorbox}[colback=white,sharp corners]
\begin{description}
\item[Problem:] \minvci
\item[Input:] A graph $G=(V, E)$, a non-negative integer $t$ and a vertex subset $B\subseteq V$
\item[Question:] Does there exist a set $X \subseteq B$ such that $|X|\leq t$ and $X\cap C\neq \emptyset$ for every minimum vertex cover $C$ of $G$?
\end{description}
\end{tcolorbox}

The difference from \bminvci is that we do not have a restriction that $G$ is bipartite.

For \minvci, we prove the following theorem.

\begin{theorem}
    \label{thm:minvci-sigma2p}
    \minvci is $\Sigma_2^p$-complete.
\end{theorem}
\begin{proof}
    The proof consists of two parts: (1) the membership in $\Sigma_2^p$ and (2) the $\Sigma_2^p$-hardness.

We first prove (1)~the membership in $\Sigma_2^p$.
To do so, we design a $\Sigma_2^p$-algorithm for \minvci.
The algorithm is based on the observation that an instance $(G, t, B)$ is a yes-instance if and only if
there exist $X \subseteq B$ and a vertex cover $C^*$ of $G$ such that 
$|X| \le t$ and $X \cap C \neq \emptyset$ for any vertex cover $C$ with $|C| \le |C^*|$. 
Imagine that we act as an $\exists$-player, and the adversary acts as a $\forall$-player.
We take the first turn, and the adversary takes the second turn.

We first guess a set $X \subseteq B$ that is supposed to be a solution to \minvci and also guess a (minimum) vertex cover $C^*$ of $G$.
Then, the adversary needs to verify that
\begin{enumerate}
    \item $|X|\leq t$, and
    \item $C^*$ is a vertex cover of $G$.
    \item $X\cap C \neq \emptyset$ for every minimum vertex cover $C$ of $G$.
\end{enumerate}
The first and second conditions are easy to verify in polynomial time.
For the third condition, the adversary selects an arbitrary vertex cover $C$ of $G$.
If $C$ is a minimum vertex cover, then $|C| \leq |C^*|$.
When this inequality holds, we check $X\cap C\neq \emptyset$.
After the adversarial choice of $C$, it only takes polynomial time.
We note that we do not have to guess a \emph{minimum} vertex cover $C^*$. 
However, if we do not guess a minimum vertex cover, the situation becomes more severe for us since the adversary can refute more easily.
This finishes the proof of (1) the membership in $\Sigma_2^p$.

We then prove (2)~the $\Sigma_2^p$-hardness.
Our proof indeed uses the same reduction by Gr\"une and Wulf~\cite{DBLP:conf/ipco/GruneW25} for the $\Sigma_2^p$-completeness of \vci, defined as follows.
\begin{tcolorbox}[colback=white,sharp corners]
\begin{description}
\item[Problem:] \vci
\item[Input:] A graph $G=(V, E)$, non-negative integers $t, k$ and a vertex subset $B\subseteq V$
\item[Question:] Does there exist a set $X \subseteq B$ such that $|X|\leq t$ and $X\cap C\neq \emptyset$ for every vertex cover $C$ of $G$ of size at most $k$?
\end{description}
\end{tcolorbox}

Their reduction used the combinatorial interdiction $3$-SAT\@, which is $\Sigma_2^p$-complete.
\red{Namely, the reduction created an instance \vci from an instance of the combinatorial interdiction $3$-SAT\@. }
It can be observed that a given instance of the combinatorial interdiction $3$-SAT is a yes-instance if and only if the created instance by the reduction of \vci is a yes-instance \red{with $k$ equal to the minimum vertex cover size}.
\red{Hence, their reduction can be treated as one from the combinatorial interdiction $3$-SAT to \vci, and thus} we conclude that \minvci is $\Sigma_2^p$-hard.
\end{proof}

\end{document}